\documentclass[11pt]{article}
\usepackage{path,graphicx}
\usepackage{a4,amssymb, dsfont}
\usepackage{authblk}
\usepackage{amsmath}
\usepackage{lineno}
\usepackage{enumitem}
\usepackage{cite}
\usepackage[]{algorithm2e}
\usepackage[margin=1.07in]{geometry}
\usepackage[T1]{fontenc}
\usepackage[utf8]{inputenc}
\usepackage{authblk}
\usepackage[justification=centering]{caption}
\newtheorem{definition}{Definition}[section]
\newtheorem{lemma}[definition]{Lemma}
\newtheorem{theorem}[definition]{Theorem}
\newtheorem{corollary}[definition]{Corollary}

\newtheorem{claim}[definition]{Claim}

\newtheorem{remark}[definition]{Remark}

\newcommand{\qed}{\hfill $\square$\smallskip}

\renewcommand{\S}{L}
\newcommand{\Sp}{S_+}
\newcommand{\Sm}{S_-}
\title{Phase Transition in Democratic Opinion Dynamics}
\author{Bernd G\"artner\thanks{gaertner@inf.ethz.ch}}
\author{Ahad N. Zehmakan\thanks{abdolahad.noori@inf.ethz.ch }}
\affil{Department of Computer Science, ETH Zurich}
\providecommand{\keywords}[1]{\textbf{\textit{Key Words:}} #1}
\date{} 
\begin{document}
\maketitle
\begin{abstract}
Consider a community where initially, each individual is positive or negative regarding a reform proposal. In each round, individuals gather randomly in fixed rooms of different sizes, and all individuals in a room agree on the majority opinion in the room (with ties broken in favor of the negative opinion). The \emph{Galam model}---introduced in statistical physics, specifically sociophysics---approximates this basic random process. We approach the model from a more mathematical perspective and study the threshold behavior and the consensus time of the model. 
\end{abstract}

\section{Introduction}
\label{Introduction}
Different discrete-time dynamic processes have been introduced to model various social phenomena like rumor spreading~\cite{chierichetti2009rumor,wang2017rumor,merlone2014reaching}, opinion forming~\cite{gartner2017color,zehmakan2019tight,zehmakan2018opinion}, and fear propagation~\cite{adler1991bootstrap}. In the present paper, we focus on the model introduced by Galam~\cite{galam2002minority}, also known as \emph{democratic opinion dynamics in random geometry}.

In the real world, opinion forming in a community is a complex process that cannot be explained in purely mathematical terms. There are recurring patterns, however, that one can try to capture mathematically. One such pattern is that an initial minority is able to eventually win the debate ---an example is the Irish ``No'' to the Nice European treaty~\cite{galam2002minority}. In this case, the social sciences offer possible explanations, for example based on \emph{the spiral of silence}~\cite{Noelle-Neumann,moscovici1991silent}. In contrast, Galam has investigated whether there are also some purely mathematical mechanisms that could (partially) explain how an initial minority is able to ``take over''~\cite{galam2002minority}. 

By now, the Galam model is a very well-established model in statistical physics, particularly sociophysics~\cite{galam2008sociophysics}, and there is a large body of literature about it~\cite{galam2004contrarian,galam2008sociophysics,gekle2005opinion,galam2005heterogeneous,qian2015activeness}. Prior work, mostly done by physicists, investigates the model experimentally, or analytically (in specific settings that are simple enough to allow for analytic solutions).

What is lacking so far is a sound theory of the Galam model in its full generality. In the current paper, we provide such a theory. We formally state and prove several basic properties of the model that were experimentally observed before, or have analytically been proven in special cases. We also give results concerning the \emph{consensus time} and the \emph{threshold behavior} of the Galam model: How long does it take to win the debate in the Galam model? And how sensitive is the outcome, with respect to the initial opinions?

\paragraph{The room-wise majority model.}
Consider a society of size $n$ where initially $n_+(0)$ individuals are positive regarding a reform proposal, and the other ones are against it. In discrete-time rounds $t=1,2,\ldots$, individuals gather randomly in fixed rooms of different sizes from $2$ to a constant $\S\ge 2$, where the numbers of seats in all rooms sum up to $n$. After each round, all individuals in a room have adopted the majority opinion in the room. At a tie, this will be the negative opinion; we call this the \emph{majority rule}. Rooms are meant to model physical spaces such as offices, houses, bars, and restaurants; time rounds correspond to social interaction at lunches, dinners, parties etc. We are interested in the sequence $n_+(1), n_+(2),\ldots$, where $n_+(t)$ (a random variable) denotes the number of positive individuals after round $t$. Similarly, we define $n_-(t)$.

The tie-breaking rule can be thought of as a social \textit{inertia principle}, where in an unclear situation, individuals tend to prefer the status quo. Hence, we have a bias towards the negative opinion built into the system. This bias can actually be quite strong: if all rooms are of size $2$, one negative individual suffices to eventually make the whole society adopt the negative opinion. On the other hand, if all rooms are of odd size, there is no bias. In general (rooms of various sizes) it is interesting to quantify the bias, and this is one thing that we are after here.

\paragraph{The Galam model.}
If we randomly match individuals to seats in round $t$, any fixed seat will be occupied by a positive individual with probability $P_{+}(t-1):={n_+(t-1)}/{n}$. However, for a set of seats, these events are not independent, which complicates the analysis. Galam~\cite{galam2002minority} therefore suggested the following model which approximates the room-wise majority model and is called \emph{Galam model}. In this model, we talk about seats, not individuals, and it works as follows: In round $t$, each seat is made positive with probability $P_{+}(t-1)$, independently from the other seats, and negative with probability $P_-(t-1):=n_-(t-1)/n=1-P_{+}(t-1)$. Once each seat has been assigned an opinion, we proceed as in the room-wise majority model to determine the number of positive seats $n_{+}(t)$.

The independent assignment of opinions to seats will in general of course result in a number of positive seats different from $n_+(t-1)$; in expectation, the number of positive seats is still $n_+(t-1)$, though, and this is why we call the Galam model an approximation of the room-wise majority model. But we explicitly point out that we do not assess the quality of this approximation here. In this paper, we consider the Galam model to be the ground truth. 

The Galam model defines a Markov chain on the state space $\{0,1,\ldots,n\}$, where state $u$ corresponds to having $u$ positive seats. There are two absorbing states $\Sm:=0$ and $\Sp:=n$. Moreover, for any state $u\notin \{\Sm,\Sp\}$, there is a positive probability of going to state $\Sm$ or $\Sp$ in the next step (if all seats turn out to have the same opinion). Hence, the Markov chain eventually converges to $\Sm$ or $\Sp$. We want to understand how quickly this happens, and---more importantly---how likely it is that state $\Sm$, say, is reached, depending on the starting state $n_+(0)$.

The chain is specified by the transition probabilities $\pi_{uv}=\Pr[n_{+}(t)=v|n_{+}(t-1)=u]$. The following definition formally introduces these transition probabilities.

\begin{definition}[Galam process]\label{def:Galam}
Let $\S$ (maximal room size) and $n$ (number of seats) be natural numbers. For a natural number $i\in[\S]$ (room size) and $u\in[n]$ (total number of positive seats after the previous round), define
\[
\pi (i,u) := \sum_{j=\lfloor i/2\rfloor+1}^{i}\binom{i}{j}\left(\frac{u}{n}\right)^j\left(1-\frac{u}{n}\right)^{i-j}
\]
(probability of the room being positive in the end of the round).

Let ${\cal R}$ (rooms) be a partition of $[n]:=\{1,\ldots,n\}$ such that $1< |R|\leq \S$ for all $R\in{\cal R}$. For a natural number $v\in[n]$ (number of positive seats by the end of the round), let ${\cal R}_v\subseteq 2^{\cal R}$ be the set of all subsets of rooms with a total of $v$ seats,
\[
{\cal R}_v = \{{\cal Q}\subseteq {\cal R}: |\cup_{Q\in {\cal Q}} Q|=v\}.
\]

For natural numbers $u,v\in[n]$, define
\[
\pi_{uv} := \sum_{{\cal Q}\in {\cal R}_v} \prod_{R\in {\cal Q}} \pi(|R|,u) \prod_{R\in {\cal R}\setminus {\cal Q}}(1- \pi(|R|,u))
\] 
(probability of having exactly $v$ positive seats by the end of the round). For fixed $n_+(0)\in[n]$, the \emph{Galam process} is the Markov chain (sequence of random variables) $n_+(t), t\geq 1$, defined by 
\[
\Pr[n_{+}(t)=v|n_{+}(t-1)=u] = \pi_{uv}, \quad u,v\in[n],~ t\geq 1.
\]
\end{definition}

\paragraph{Statement of results.}
Galam~\cite{galam2002minority} conjectured that the Galam process exhibits a threshold behavior with one phase transition; i.e., there is a threshold value $\alpha$ so that $P_+(0)<\alpha$ and $P_+(0)>\alpha$ result in $S_{-}$ and $S_{+}$, respectively, with probability approaching one. Moreover, he claimed that if all rooms are of odd size, then the threshold value $\alpha$ is equal to $1/2$ because there is no tie-breaking, otherwise $\alpha>1/2$, meaning that the negative opinion has an advantage.

In the present paper, we prove the following results.

\begin{itemize}
\item[(i)] If all rooms are of the same size $i=\S\geq 3$, there is a threshold value $\alpha_i$ such that for every $\epsilon>0$, $P_+(0)\geq \alpha_i+\epsilon$ results in convergence to $\Sp$, and $P_+(0)\leq \alpha_i-\epsilon$ in convergence to $\Sm$, asymptotically almost surely (a.a.s.)\footnote{We say an event occurs asymptotically almost surely (a.a.s.) if its probability is at least $1-o(1)$ as a function of $n$.}. In both cases, the \emph{consensus time} (expected time to convergence) is $\mathcal{O}(\log\log n)$, and this bound is best possible.

\item[(ii)] The same convergence behavior holds if rooms have different sizes $3\leq i\leq \S\leq 16$, with a threshold value $\alpha$ depending on the distribution of room sizes. The restriction of $\S\leq 16$ comes from the limitations of our proof technique; the result probably holds also for $\S>16$. 

\item[(iii)] If all rooms are of size $i=\S=2$, the consensus time is $\mathcal{O}(\log n)$, and this bound is best possible. Qualitatively, this case is different from the other ones, as there is no threshold: any fixed value $P_+(0)<1$ will still lead to convergence to $\Sm$, a.a.s.

\item[(iv)] We prove that for even $i$, $\alpha_{i+2}<\alpha_i$, meaning that the threshold for the positive opinion to win the debate goes down if rooms get larger. This is plausible, since ties (that benefit the negative opinion) become less likely.
\end{itemize}

We therefore rigorously confirm Galam's conjecture, and we also determine how quickly the debate is won in the Galam model. The bound of $\mathcal{O}(\log\log n)$ on the consensus time in the case (i) shows that this happens very quickly, asymptotically. The constant hidden in this bound depends on how close we are to the threshold.

After the introduction of the Galam model~\cite{galam2002minority}, a series of extensions have been suggested to make the model more realistic. Some of those modifications include adding contrarian effect \cite{galam2004contrarian}, introducing a random tie-breaking rule \cite{galam2005heterogeneous}, considering three competing opinions \cite{gekle2005opinion}, inflexible individuals \cite{galam2007role}, and defining the level of activeness \cite{qian2015activeness}. We briefly introduce some of these variants in Section \ref{conclusion} and discuss the possibility of extending our techniques to analyze these variants as prospective research work.
 
To summarize, we are dealing with the Galam model---which is a very well-studied dynamic opinion forming model in the literature of sociophysics---from a more theoretical angle. Our approach not only allows us to prove some conjectures about the behavior of the process, it also provides us with mathematical tools to achieve some new interesting insights regarding the consensus time and threshold behavior of the process. In the bigger picture, the main goal of this research is to make one of the first steps in the direction of building a new bridge between the study of the opinion forming dynamics in sociophysics and theoretical computer science.

\section{Preliminaries}

\paragraph{The expected transition probability.} As the formula for the transition probabilities $\pi_{uv}$ in Definition~\ref{def:Galam} might suggest, the Galam process is rather complicated in full detail. But we will show in the next section that conditioned on the current state, the next state is sharply concentrated around its mean in the relevant cases. We therefore start by computing this mean. As we are trying to understand the asymptotic behavior for $n\rightarrow\infty$, it turns out to be more convenient to work with the positive seat probabilities $P_+(t) = n_+(t)/n$ instead of the actual numbers $n_+(t)$ of positive seats.

\begin{definition}
With $n,u,v$ and $\pi_{uv}$ as in Definition~\ref{def:Galam}, and for $p,q\in[0,1]$, we define $P_+^p$
to be the random variable with distribution $\Pr[P_+^p = q] = \pi_{pn, qn}$.
\end{definition}
In words, $P_+^p$ tells us how the positive seat probability evolves in one round of the Galam process, if that probability is $p$ at the beginning of the round. Note that formally, this is only defined if $p$ and $q$ are integer multiples of $1/n$, but as we are interested in the situation $n\rightarrow\infty$, we typically think of $p$ and $q$ as arbitrary probabilities. Furthermore, we define $P_-^p:=1-P_+^{1-p}$, which tells us how the negative seat probability evolves in one round of the Galam process, if that probability is $p$ at the beginning of the round. We often argue about $P_+^p$ and $P_-^p$ simultaneously in which case we use $P_{\pm}^p$ to stand for both random variables. Let us also define $n_{\pm}^p:=nP_{\pm}^p$.

\begin{lemma}
Suppose there are $r_i$ rooms with $i$ seats for $1< i \leq \S$ ($r_{\S}\ne 0$), and that $\sum_{i=1}^{\S}a_i=1$, where $a_i:={ir_i}/{n}$ is the fraction of seats in rooms of size $i$. With $\pi(\cdot,\cdot)$ as in Definition~\ref{def:Galam}, we have
\begin{equation}
\label{eq 1}
\mathbb{E}[P_+^p]=\sum_{i=1}^{\S}a_i~ \pi(i, np) = \sum_{i=1}^{\S}a_i \sum_{j=\lfloor\frac{i}{2}\rfloor+1}^{i} {i \choose j} p^{j}(1-p)^{i-j}.
\end{equation}
\end{lemma}
To prove this, we use linearity of expectation on top of the fact that a room of size $i$ contributes $i\pi(i,np)$ positive seats on expectation.

Furthermore, we have
\begin{equation}
\label{eq2}
\mathbb{E}[P_-^p]= \sum_{i=1}^{\S}a_i \sum_{j=\lceil i/2\rceil}^{i} {i \choose j} p^j(1-p)^{i-j}
\end{equation}
which should be clear since the second sum is the probability that the outcome of a room of size $i$ is negative when each seat is occupied by a negative individual with probability $p$, independently. 


Equations~(\ref{eq 1}) and~(\ref{eq2}) suggest to analyze the involved polynomials.

\begin{definition}
\label{definition:function}
For a distribution $a_1,\ldots,a_{\S}$ that we consider to be fixed, we define polynomials
\[
f_i(p):=\sum_{j=\lfloor i/2 \rfloor+1}^{i} {i \choose j} p^j (1-p)^{i-j}, \quad i\in[\S], p\in[0,1]
\]
as well as
\[
f(p) :=\sum_{i=1}^{\S}a_if_i(p) = \mathbb{E}[P_+^p].
\]
 Furthermore, we define the functions $h_i(p)$ and $h(p)$ respectively to be $f_i(p)-p$ and $f(p)-p$.
\end{definition} 

The main idea behind our proof regarding the phase transition is as follows. We show that function $f(p)$ has a unique fixed point $\alpha$ in the interval $(0,1)$ such that $f(p)<p$ for $p\in (0,\alpha)$ and $f(p)>p$ for $p\in (\alpha,1)$ (see Figure \ref{fig1}). Furthermore, it is proven that $P_+^p$ is sharply concentrated around its expectation, which yields our desired threshold behavior.

\begin{figure}[h]
\begin{center}
\includegraphics[width=0.3\textwidth]{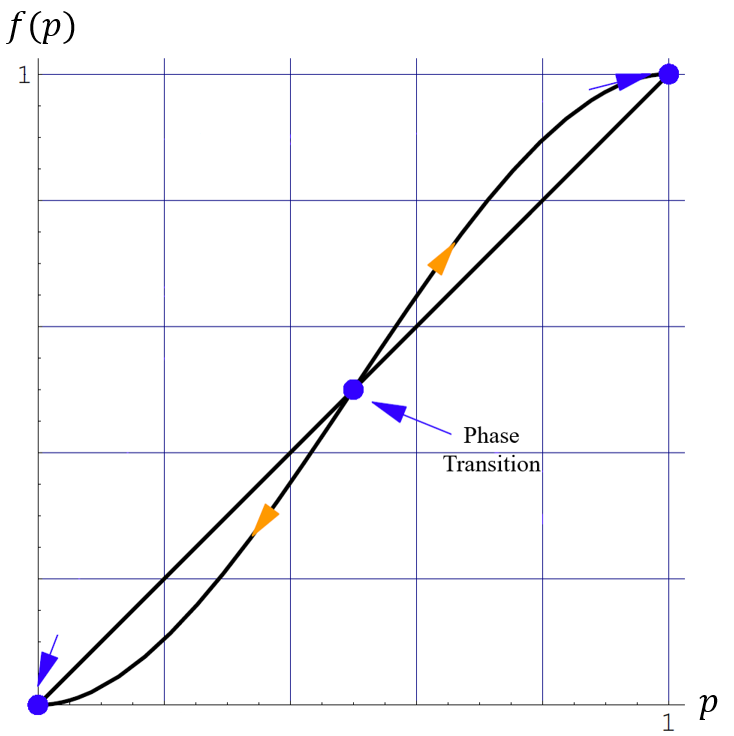}
\caption{Variation of function $f(p)$ for $p\in[0,1]$\label{fig1}}
\end{center}
\end{figure}

Galam~\cite{galam2002minority} already computed the unique fixed points of the functions $f_i(p)$ for $3\leq i\leq 6$, which are equal to $1/2$ for $i=3,5$, and $\frac{1+\sqrt{13}}{6}\approx 0.77$ and approximately $0.65$ respectively for $i=4$ and $i=6$. 
 
\section{Phase Transition and Consensus Time}
\label{Galam Model}
We first set up some concentration result on the random variables $P_{\pm}^p$ in Section~\ref{ingredients}. Building on that, we discuss the threshold behavior and the consensus time of the Galam process in Sections~\ref{room-size-3},~\ref{general-case}, and~\ref{room-size-2}, for different settings. Furthermore, we draw the connection between the threshold value and the room sizes, in Section~\ref{threshold-room}.

\subsection{Ingredients: Tail Bounds}
\label{ingredients}
By applying Azuma's inequality \cite{feller1968introduction} (see Theorem \ref{theorem 6}), we provide Corollary~\ref{lemma 6} which asserts that the random variables $P_{\pm}^p$ are sharply concentrated around their expectations. This allows us to prove Lemma~\ref{lemma 5}, concerning how the positive/negative seat probability evolves during the process, which is the main ingredient of our proofs regarding the threshold behavior of the process. 

If we are given $n$ discrete probability spaces $(\Omega_i,Pr_i)$ for $1 \leq i \leq n$, then their product is defined to be the probability space over the ground set $\Omega:=\Omega_1 \times \Omega_2 \times \dots \times \Omega_n$ with the probability function 
\[
\Pr[(\omega_1, \dots, \omega_n)]:=\prod_{i=1}^{n}Pr_i[\omega_i],
\]
where $\omega_i \in \Omega_i$. This simply means that a random element $\omega=(\omega_1, \dots,\omega_n) \in \Omega$ is obtained by drawing the components $\omega_i \in \Omega_i$ independently according to their respective probability distributions $Pr_i$.

Let $(\Omega,Pr)$ be the product of $n$ discrete probability spaces, and let $X:\Omega\rightarrow \mathbb{R}$ be a random variable over $\Omega$. We say that the \emph{effect} of the $i$-th coordinate is at most $c_i$ if for all $\omega, \omega' \in \Omega$ which differ only in the $i$-th coordinate we have
\[
|X(\omega)-X(\omega')|\leq c_i.
\]
Azuma's inequality states that $X$ is sharply concentrated around its expectation if the effect of the individual coordinates is not too big.

\begin{theorem}
\label{theorem 6}
(Azuma's inequality \cite{feller1968introduction}) Let $(\Omega,Pr)$ be the product of $n$ discrete probability spaces $(\Omega_i,Pr_i)$ for $1 \leq i \leq n$, and let $X:\Omega\rightarrow \mathbb{R}$ be a random variable with the property that the effect of the $i$-th coordinate is at most $c_i$. Then, for all $\delta> 0$ we have
\[
\Pr[(1-\delta)\mathbb{E}[X]< X < (1+\delta)\mathbb{E}[X]]\geq 1-2\exp(-\frac{\delta^2\mathbb{E}[X]^2}{2\sum_{i=1}^{n}c_{i}^2}).
\]
\end{theorem} 

\begin{corollary}
\label{lemma 6}
In the Galam model, for all $\delta>0$
\[
\Pr[(1-\delta)\mathbb{E}[P_{\pm}^p]<P_{\pm}^p<(1+\delta)\mathbb{E}[P_{\pm}^p]]\geq 1-\exp(-\Theta(n\mathbb{E}[P_{\pm}^p]^2)).
\]
\end{corollary}
\begin{proof}
We prove the statement for $P_+^p$ and the case of $P_-^p$ can be proven analogously. Label the seats from $1$ to $n$. Corresponding to the $i$-th seat for $1\le i\le n$, we define the probability space $(\Omega_i,Pr_i)$. Now, $P_+^p:\Omega \rightarrow \mathbb{R}$ is defined over $\Omega$, where $(\Omega,Pr)$ is the product of $(\Omega_i,Pr_i)$. Consider two elements $\omega,\omega'\in\Omega$ which differ only in the $i$-th coordinate (i.e., two opinion assignments which disagree only on the opinion of the $i$-th seat). Since the size of a room is at most $L$, we have
\[
|P_+^p(\omega)-P_+^p(\omega')|\le \frac{L}{n}.
\] 
Hence, by applying Azuma's inequality (see Theorem \ref{theorem 6}), we have for all $\delta>0$
\[
\Pr[(1-\delta)\mathbb{E}[P_+^p]<P_+^p<(1+\delta)\mathbb{E}[P_+^p]]\geq 1-2\exp(-\frac{\delta^2\mathbb{E}[P_+^p]^2}{\frac{2n \S^2}{n^2}})=1-\exp(-\Theta(n\mathbb{E}[P_+^p]^2))
\]
since both $\S$ and $\delta$ are constant in our context. \qed
\end{proof}

Now, we present Lemma~\ref{lemma 5}, which will be needed several times later on. The goal here is the following: we want to understand the evolution of the Galam process over a number of rounds, i.e.\ we want to argue about the distribution of $n_{\pm}(t+t_0)$, given $n_{\pm}(t_0)$. For example, if $n_+(t_0)$ is already large, we expect $n_+(t+t_0)$ to be significantly larger (a sufficient majority will expand its expected share). It will usually be easy to prove lower or upper bounds for the expected drift of the process in one round. Formally, these will be bounds of the form $Q(p)\le \mathbb{E}[P_{\pm}^p]$, or $\mathbb{E}[P_{\pm}^p]\le R(P)$ that hold for $p$ in a suitable range.

From this, we want to conclude that (a) this drift happens a.a.s.\ (for this we can already use Azuma's inequality), and (b) that it happens over a number of consecutive rounds a.a.s.\ This is what Lemma~\ref{lemma 5} will do for us. On a qualitative level, conclusion (b) seems clear, but we still need to provide a quantitative version that is parameterized with the expected drift in one round.

\begin{lemma}
\label{lemma 5}
Let $\delta>0$ be a fixed real number, and let $I$ be a monotone decreasing function such that for all $p\in[0,1]$, $(1-\delta)\mathbb{E}[P_{\pm}^p]< P_{\pm}^p< (1+\delta)\mathbb{E}[P_{\pm}^p]$ with probability at least $1-I(\mathbb{E}[P_{\pm}^p])$. (Note that Corollary~\ref{lemma 6} gives us such a function.) Let $[\underline{b},\overline{b}]\subseteq [0,1]$ be a fixed interval. 

\begin{itemize}
\item[(i)] \textbf{Lower bound on the drift over $\textbf{m}$ rounds}:
Let $Q$ be a nonnegative and monotone increasing function such that for all $p\in [\underline{b},\overline{b}]$, $\mathbb{E}[P_{\pm}^p]\geq Q(p)$. We define functions $Q^t$ over $[0,1]$ as follows.
\begin{eqnarray*}
Q^0 (p) &:=& p, \\
Q^t (p) &:=& (1-\delta)Q (Q^{t-1}(p)), \quad t>0.
\end{eqnarray*}
Suppose for some $b\in [\underline{b},\overline{b}]$ and some $m\ge 1$, we have $Q^t(b)\in [\underline{b},\overline{b}]$ for all $t<m$. Then the probability that $P_{\pm}(t'+m) \geq Q^{m}(b)$ given $P_{\pm}(t')\geq b$ for some $t'\ge 0$, is at least $1-\sum_{t=1}^{m}I(Q^{t}(b))$. 
\item[(ii)] \textbf{Upper bound on the drift over $\textbf{m}$ rounds}:
Let $R$ be a nonnegative and monotone increasing function such that for all $p\in [\underline{b},\overline{b}]$, $\mathbb{E}[P_{\pm}^p]\leq R(p)$. We define functions $R^t$ over $[0,1]$ as follows.
\begin{eqnarray*}
R^0 (p) &:=& p, \\
R^t (p) &:=& (1+\delta)R (R^{t-1}(p)), \quad t>0.
\end{eqnarray*}
Suppose for some $b\in[\underline{b},\overline{b}]$ and some $m\ge 1$, we have $R^t(b)\in [\underline{b},\overline{b}]$ for all $t<m$. Then the probability that $P_{\pm}(t'+m) \leq R^{m}(b)$ given $P_{\pm}(t')\le b$ for some $t'\ge 0$, is at least $1-\sum_{t=1}^{m}I(R^{t}(b))$.

\item[(iii)] For $Q(p)=Kp^{\ell}$ and $R(p)=Kp^{\ell}$, where $K,\ell$ are some constants, we have \[
Q^{t}(b)\ge (K(1-\delta))^{\sum_{j=0}^{t-1}\ell^{j}}b^{\ell^t}, \quad t\ge 1,
\]
and
\[
R^t(b)\le (K(1+\delta))^{\sum_{j=0}^{t-1}\ell^{j}}b^{\ell^t}, \quad t\ge 1.
\]
\end{itemize}
\end{lemma}

\begin{proof}
We prove part (i) where we w.l.o.g.\ assume that $t'=0$, and part (ii) is analogous. Let us define events
\[
E_ i = \{P_{\pm}(i)\ge Q^{i}(b)\}, \quad i\geq 0.
\]
We are interested in bounding $\Pr[E_m|E_0] = \Pr[E_m\cap E_0] / \Pr[E_0]$. We have that
\begin{eqnarray*}
\Pr[E_m\cap E_0] &\ge& \Pr[\bigcap_{i=0}^mE_i] = \prod_{t=0}^m\frac{\Pr[\bigcap_{i=0}^tE_i]}{\Pr[\bigcap_{i=0}^{t-1}E_i]}
= \prod_{t=0}^m\frac{\Pr[E_t \cap \bigcap_{i=0}^{t-1}E_i]}{\Pr[\cap_{i=0}^{t-1}E_i]} \\
&=& \prod_{t=0}^m\Pr[E_t| \bigcap_{i=0}^{t-1}E_i] = \Pr[E_0]\cdot \prod_{t=1}^m\Pr[E_t| \bigcap_{i=0}^{t-1}E_i].
\end{eqnarray*}
Hence,
\[
\Pr[E_m|E_0] \ge \prod_{t=1}^m\Pr[E_t| \bigcap_{i=0}^{t-1}E_i] = \prod_{t=1}^m\Pr[E_t| E_{t-1}],
\]
since $P_{\pm}(t)$ only depends on $P_{\pm}(t-1)$. Because $\prod_{t=1}^m (1-I(Q^t(b))) \geq 1 - \sum_{t=1}^m I(Q^t(b))$, the statement is implied by the following 
\begin{claim}
\label{claim 1}
$\Pr[E_t| E_{t-1}] \geq 1-I(Q^t(b))$.
\end{claim}

\paragraph{Proof of Claim~\ref{claim 1}.} We show that for every $p\geq Q^{t-1}(b)$, 
\[
\\Pr[P_{\pm}(t)\ge Q^{t}(b) | P_{\pm}(t-1) = p] \geq 1-I(Q^t(b)). 
\]
A simple calculation (omitted) then shows that the same bound also holds for 
\[
\Pr[E_t| E_{t-1}] = \\Pr[P_{\pm}(t)\ge Q^{t}(b) | P_{\pm}(t-1) \geq Q^{t-1}(b)].
\]

To prove the claim, we exploit that 
\begin{equation}\label{eq:chain}
Q^t(b) = (1-\delta)Q(Q^{t-1}(b)) \leq (1-\delta)\mathbb{E}[P_{\pm}^{Q^{t-1}(b)}) \leq (1-\delta) \mathbb{E}[P_{\pm}^p] \leq \mathbb{E}[P_{\pm}^p].
\end{equation}
The first inequality holds since $Q^{t-1}(b)\in[\underline{b},\overline{b}]$ by assumption. The second inequality is monotonicity of the function $\mathbb{E}[P_{\pm}^{(\cdot)}]$, along with $p\geq Q^{t-1}(b)$. Hence, we can infer that
\begin{eqnarray*}
\\Pr[P_{\pm}(t)\ge Q^{t}(b) | P_{\pm}(t-1) = p] &\geq& \\Pr[P_{\pm}(t)\ge (1-\delta) \mathbb{E}[P_{\pm}^p] | P_{\pm}(t-1) = p] \\
&=:& \Pr [P_{\pm}^p\ge (1-\delta) \mathbb{E}[P_{\pm}^p]] \geq 1 - I(\mathbb{E}[P_{\pm}^p]),
\end{eqnarray*}
by applying the definition of $P_{\pm}^p$ and then the assumption. The claim then follows from $I$ being a decreasing function and the same chain of inequalities as in Equation~(\ref{eq:chain}).

Part (iii) is proven by induction on $t$ for $Q^{t}(b)$. The same proof applies to the case of $R^{t}(b)$. For $t=1$, $Q^1(b)=(1-\delta)Q(Q^0(b))=(1-\delta)Q(b)=(1-\delta)Kb^{\ell}$. Let the statement hold for $t-1\geq 1$. Since $Q^{t}(b)=(1-\delta)Q(Q^{t-1}(b))$, by applying the induction hypothesis and the fact that $Q$ is increasing we have
\begin{align*}
\centering
& Q^t(b)\ge (1-\delta)Q((K(1-\delta))^{\sum_{j=0}^{t-2}\ell^j}b^{\ell^{t-1}})=(1-\delta)K((K(1-\delta))^{\sum_{j=0}^{t-2}\ell^j}b^{\ell^{t-1}})^\ell=\\
& (1-\delta)K(K(1-\delta))^{\sum_{j=1}^{t-1}\ell^j}b^{\ell^{t}}=(K(1-\delta))^{\sum_{j=0}^{t-1}\ell^{j}}b^{\ell^t}.
\end{align*}\qed

\end{proof}

\subsection{All Rooms of Size $i\ge 3$}
\label{room-size-3}
In this section, we prove that the Galam process exhibits a threshold behavior with one phase transition if all rooms are of the same size $i\geq 3$. That is, there is a threshold value $\alpha_i$ such that if $P_+(0)$ is ``slightly'' more (less, respectively) than $\alpha_i$, then the process converges to $\Sp$ ($\Sm$, respectively). 
\begin{theorem}
\label{theorem 5} 
In the Galam model with $a_i=1$ (i.e., all rooms are of the same size $i$) for some $i\ge 3$ and for any $\epsilon>0$,

(i) $P_{+}(0)\geq \alpha_i+\epsilon$ results in $\Sp$

(ii) $P_{+}(0)\leq \alpha_i-\epsilon$ results in $\Sm$
\\
in $\mathcal{O}(\log\log n)$ rounds a.a.s., where $\alpha_i$ is the unique fixed point of $f_i(p)$ in $(0,1)$.
\end{theorem}
First in Lemma~\ref{lemma 3}, we prove that function $f_i(p)$, which is equivalent to $\mathbb{E}[P_+^p]$ in this setting (see Definition~\ref{definition:function}), has a unique fixed point $\alpha_i$; that is, $f_i(p)<p$ for $p\in(0,\alpha_i)$ and $f_i(p)>p$ for $p\in(\alpha_i,1)$. Therefore, by starting from some positive seat density $p$ larger (smaller, respectively) than $\alpha_i$, we expect that the positive seat probability increases (decreases, respectively) in each round. Intuitively, this explains the aforementioned threshold behavior. However, to make an a.a.s. statement and determine the consensus time of the process, we need to apply the concentration results from Section~\ref{ingredients} and do some careful calculations. 

\begin{lemma}
\label{lemma 3}
For $i\ge 3$, function $f_i(p)$ has a unique fixed point $\alpha_i$ in $(0,1)$ such that $f_i(p)<p$ for $p\in(0,\alpha_i)$ and $f_i(p)>p$ for $p\in(\alpha_i,1)$.
\end{lemma}
\begin{proof} 
It suffices to prove that function $h_i(p)=f_i(p)-p$ has a unique root $\alpha_i$ in $(0,1)$ such that $h_i(p)$ is negative and positive respectively for $p\in(0,\alpha_i)$ and $p\in(\alpha_i,1)$. First by applying induction, we show that the first derivative of function $f_{i,l}(p):=\sum_{j=l}^{i} {i \choose j} p^j (1-p)^{i-j}$ for $1 \leq l \leq i$ is
\begin{equation}
\label{eq 2}
\frac{d f_{i,l}(p)}{d p}=l {i \choose l} p^{l-1} (1-p)^{i-l}.
\end{equation}

We do backwards induction on $l$. Consider the base case of $l=i$. The first derivative of $f_{i,i}=p^i$ is equal to $ip^{i-1}$ which satisfies Equation~(\ref{eq 2}). Now as the induction hypothesis, assume for some $1<l\leq i$, Equation~(\ref{eq 2}) holds. We prove that it holds also for $l-1$; i.e., $\frac{d f_{i,l-1}(p)}{d p}=(l-1) {i \choose l-1} p^{l-2} (1-p)^{i-l+1}$. By definition of $f_{i,l}(p)$,
\[
\frac{d f_{i,l-1}(p)}{d p}=\frac{d}{d p}({i \choose l-1}p^{l-1}(1-p)^{i-l+1})+\frac{d f_{i,l}(p)}{d p}.
\] 
Thus, by the induction hypothesis $\frac{d f_{i,l-1}(p)}{d p}$ is equal to
\begin{eqnarray*}
&&{i \choose l-1}(l-1)p^{l-2}(1-p)^{i-l+1}-{i \choose l-1}(i-l+1)p^{l-1}(1-p)^{i-l} + l{i \choose l}p^{l-1}(1-p)^{i-l} \\
&=& {i \choose l-1}(l-1)p^{l-2}(1-p)^{i-l+1},
\end{eqnarray*}
because ${i \choose l-1}(i-l+1)={i \choose l}l$.

In particular, the first derivative of $h_i(p)=f_{i,\lfloor i/2 \rfloor+1}-p$ is equal to
\begin{equation}
\label{eq 2'}
\frac{dh_i(p)}{dp}=(\lfloor \frac{i}{2}\rfloor+1) {i \choose \lfloor\frac{i}{2}\rfloor+1} p^{\lfloor\frac{i}{2}\rfloor} (1-p)^{i-\lfloor\frac{i}{2}\rfloor-1}-1.
\end{equation}

By Equation~(\ref{eq 2'}), $\frac{dh_i(p)}{dp}$ is equal to $-1$ for both $p=0$ and $p=1$. Thus, in a very close right neighborhood of $0$, we have $h_i(p)<h_i(0)=0$ and similarly in a very close left neighborhood of $1$, we have $h_i(p)>h_i(1)=0$. This implies that there exists at least one $\alpha_i \in (0,1)$ so that $h_i(\alpha_i)=0$, because $h_i$ is a continuous function, and it also demonstrates that there exist $p_1 \in (0,\alpha_i)$ and $p_2 \in (\alpha_i,1)$ such that $h_i(p_1)<0$ and $h_i(p_2)>0$.

It only remains to show that $\alpha_i$ is unique. We show that $h_i(p)$ has exactly three roots in $[0,1]$ which must be $0$, $1$, and $\alpha_i$. Assume on the contrary, it has more than three roots and let $\alpha^{(1)} <\alpha^{(2)}<\alpha^{(3)}<\alpha^{(4)}$ be four of them such that $\alpha^{(1)}=0$ and $\alpha^{(4)}=1$. Based on Rolle's Theorem~\cite{rosenlicht1968introduction}, there exist at least three intermediate values of $\beta^{(1)}$, $\beta^{(2)}$, and $\beta^{(3)}$ such that $\alpha^{(1)} < \beta^{(1)} < \alpha^{(2)} < \beta^{(2)} < \alpha^{(3)}< \beta^{(3)} < \alpha^{(4)}$ and $\frac{d}{dp}h_i(\beta^{(1)})=\frac{d}{dp}h_i(\beta^{(2)})=\frac{d}{dp}h_i(\beta^{(3)})=0$. Now, we argue that this is not possible.

A simple calculation shows that the second derivative of $h_i(p)$ is equal to
\[
\frac{d^2h_i(p)}{dp^2}=(\lfloor\frac{i}{2}\rfloor+1)(i-1){i \choose \lfloor\frac{i}{2}\rfloor+1}p^{\lfloor\frac{i}{2}\rfloor-1}(1-p)^{i-\lfloor\frac{i}{2}\rfloor-2}(\frac{\lfloor\frac{i}{2}\rfloor}{i-1}-p).
\] 
Therefore, $\frac{dh_i(p)}{dp}$ monotonically increases from $p=0$ to $p=\frac{\lfloor\frac{i}{2}\rfloor}{i-1}$ where it achieves a maximum, and then monotonically decreases from $\frac{\lfloor\frac{i}{2}\rfloor}{i-1}$ to $p=1$. Thus, $\frac{dh_i(p)}{dp}$ can take on the value $0$ (actually any value) at most twice for $0\leq p \leq 1$, and it implies that $h_i(p)$ has exactly three roots in the interval of $[0,1]$, which are $0$, $1$, and $\alpha_i$. \qed
\end{proof}

\paragraph{Proof of Theorem~\ref{theorem 5}.}
We discuss the case of $P_+(0)\geq \alpha_i+\epsilon$; the proof of the other case is analogous. We prove by starting from $P_+(0)\geq \alpha_i+\epsilon$ a.a.s. the negative seat probability decreases to a very small constant (phase 1), then the process reaches the negative seat probability at most $1/n^{\frac{1}{4}-\delta'}$ for some small constant $\delta'>0$ in $\mathcal{O}(\log\log n)$ rounds (phase 2). Finally, it reaches $\Sp$ in three more rounds (phase 3). 

Since the formula for $\mathbb{E}[P_{+}^p]$ is quite involved (see Equation~(\ref{eq 1})), we set a lower bound $Q(p)$ of the form $Kp^{\ell}$ on $\mathbb{E}[P_{+}^p]$, which is easier to handle but only holds for $p\in[\alpha_i+\epsilon,1-\epsilon']$ for some constant $\epsilon'>0$. By applying Lemma~\ref{lemma 5} and Corollary~\ref{lemma 6}, we show that the negative seat probability decreases to $\epsilon'$ a.a.s. This argument fails for $p>1-\epsilon'$. To continue, we find an upper bound $R(p)$ on $\mathbb{E}[P_{-}^p]$ and again apply Lemma~\ref{lemma 5} and Corollary~\ref{lemma 6} to prove the negative seat probability decreases to $1/n^{\frac{1}{4}-\delta'}$, in phase 2. This only works for $p\ge 1/n^{\frac{1}{4}-\delta'}$ since otherwise the error probability attained by Corollary~\ref{lemma 6} will be a constant. Interestingly, Markov's inequality~\cite{feller1968introduction} provides us with a stronger tail bound for $p<1/n^{\frac{1}{4}-\delta'}$. This allows us to prove the process reaches $\Sp$ in at most three more rounds, which is discussed in phase 3. 

\paragraph{Phase 1.} First, we show that the process reaches the negative seat probability of at most $\epsilon'$ for an arbitrarily small constant $\epsilon'>0$ by starting from $P_+(0)\ge \alpha_i+\epsilon$ after a constant number of rounds a.a.s.

We want to apply Lemma~\ref{lemma 5} (i). Thus, first we have to determine $Q(p)$ for which $\mathbb{E}[P_{+}^{p}]\ge Q(p)$. Define function $g_i(p):=\mathbb{E}[P_+^p]/p$, which tells us how the positive seat probability evolves in one round, in expectation. Let $\rho:=\min(g_i(p))$ be the minimum value of $g_i(p)$ for $p\in[\alpha_i+\epsilon,1-\epsilon']$. Therefore, $\mathbb{E}[P_+^p]=g_i(p)p\geq\rho p$. We set $\underline{b}=\alpha_i+\epsilon$, $\overline{b}=1-\epsilon'$, $b=\alpha_i+\epsilon$, and $Q(p)=Kp^{\ell}$ for $K=\rho$ and $\ell=1$.

Furthermore by Corollary \ref{lemma 6}, $P_{+}^p\geq (1-\delta)\mathbb{E}[P_+^p]$ for an arbitrary $\delta>0$ with probability at least $1-\exp(-\Theta(n\mathbb{E}[P_+^p]^2))$. Thus, we can set $I(\mathbb{E}[P_{+}^p])=\exp(-\Theta(n\mathbb{E}[P_+^p]^2))$.

Having Claim~\ref{claim 2} (whose proof is given below) in hand, we are now ready to apply Lemma~\ref{lemma 5}.
\begin{claim}
\label{claim 2}
There is some constant integer $m\ge 1$ such that $Q^{m}(\alpha_i+\epsilon)\ge (1-\epsilon')$ and for all $t<m$ we have $Q^{t}(\alpha_i+\epsilon)\in[\alpha_i+\epsilon,1-\epsilon']$.
\end{claim}
Therefore, given $P_+(0)\ge\alpha_i+\epsilon$, we have $P_+(m)\geq Q^{m}(\alpha_i+\epsilon)\geq (1-\epsilon')$ with probability at least $1-\sum_{t=1}^{m}\exp(-\Theta(n(Q^t(\alpha_i+\epsilon))^2))$. This probability is larger than $1-\exp(-\Theta(n))$ since by Claim 2, $Q^t(\alpha_i+\epsilon)\ge \alpha_i+\epsilon$ for $1\le t\le m$ and $m$ is a constant.

Hence, starting from $P_+(0)\ge \alpha_i+\epsilon$, the process reaches the positive seat probability at least $1-\epsilon'$ in a constant number of rounds a.a.s.

\paragraph{Proof of Claim~\ref{claim 2}.} Let us first prove that $\rho$ is a constant larger than 1; recall that $\rho:=\min(g_i(p))$ for $p\in[\alpha_i+\epsilon,1-\epsilon']$, where $g_i(p)=\mathbb{E}[P_+^p]/p=f_i(p)/p$ (note $\mathbb{E}[P_+^p]=f_i(p)$ in this setting, see Definition~\ref{definition:function}). Since function $g_i(p)$ is a continuous function in the closed interval $[\alpha_i+\epsilon,1-\epsilon']$, based on Extreme Value Theorem \cite{rosenlicht1968introduction} it must attain a minimum, that is, there exists some value $p_{\min}$ such that $g_i(p_{\min})\leq g_i(p)$ for all $p\in[\alpha_i+\epsilon,1-\epsilon']$. Furthermore based on Lemma \ref{lemma 3}, we know $f_i(p)>p$ (i.e., $g_i(p)>1$) for $p\in[\alpha_i+\epsilon,1-\epsilon']$ which implies that $\rho =g_i(p_{\min})$ is equal to a constant lager than 1. 

Since $\rho$ is a constant larger than 1, we can select $\delta$ sufficiently small such that $(1-\delta)\rho$ is a constant larger than 1, which implies that 
\begin{equation}
\label{eq 5}
Q^{t}(p)=(1-\delta)\rho  Q^{t-1}(p)\ge Q^{t-1}(p)\quad \forall t\ge 1.
\end{equation}
Furthermore by part (iii) in Lemma~\ref{lemma 5}, we know that
\[
Q^{t}(\alpha_i+\epsilon)\ge (K(1-\delta))^{\sum_{j=0}^{t-1}\ell^j}b^{\ell^t}=(\rho(1-\delta))^{t}(\alpha_i+\epsilon)
\]
where we used that $\sum_{j=0}^{t-1}\ell^j=t$ and $\ell^{t}=1$ for $\ell=1$. Thus for $t'=\lceil\log_{(1-\delta)\rho}\frac{1-\epsilon'}{\alpha_i+\epsilon}\rceil$,
\begin{equation}
\label{eq 6}
Q^{t'}(\alpha_i+\epsilon)\ge 1-\epsilon'.
\end{equation} 
Putting Equation~(\ref{eq 5}) and Equation~(\ref{eq 6}) together implies that there exists some integer $m\le t'$ such that $Q^{m}(\alpha_i+\epsilon)\ge (1-\epsilon')$ and for all $t<m$ we have $Q^{t}\in[\alpha_i+\epsilon,1-\epsilon']$. It remains to prove that $m$ is a constant. We claim that $\frac{1-\epsilon'}{\alpha_i+\epsilon}\leq 2$ because $\alpha_i\geq 1/2$ (intuitively this should be clear because ties are in favor of negative opinion; however, for a formal proof please see part (i) in the proof of Theorem~\ref{theorem 9}). Thus, $m\leq t'\le \lceil \log_{(1-\delta)\rho}2\rceil$. Since $(1-\delta)\rho$ is a constant larger than 1, $m$ is also a constant.

\paragraph{Phase 2.} So far we proved that, by starting from $P_+(0)\ge \alpha_i+\epsilon$, after a constant number of rounds, say $t_0$, a.a.s. the process reaches the negative seat probability of at most $\epsilon'$, where $\epsilon'>0$ is an arbitrarily small constant. Now, we show that from $P_-(t_0)\le\epsilon'$, the process reaches the negative seat probability of at most $1/n^{\frac{1}{4}-\delta'}$ in $\mathcal{O}(\log\log n)$ rounds a.a.s. for sufficiently small constant $\delta'>0$, say $\delta'=1/23$.

Similar to phase 1, we want to apply Lemma~\ref{lemma 5}, but this time part (ii). Thus, we have to find an upper bound $R(p)$ on $\mathbb{E}[P_{-}^p]$. By Equation~(\ref{eq2}) for $a_i=1$ (i.e., all rooms are of the same size $i$), $\mathbb{E}[P_-^p]=\sum_{j=\lceil i/2\rceil}^{i}{i\choose j}p^j (1-p)^{i-j}$, which is equivalent to the probability that at least $\lceil i/2\rceil$ seats are negative in a room of size $i$ when we set each seat to be negative with probability $p$, independently. This is smaller than the probability that at least two seats are negative since $\lceil i/2\rceil\ge 2$ for $i\ge 3$. This probability is upper-bounded by ${i\choose 2}p^2\le i^2p^2$. Thus, $\mathbb{E}[P_-^p]\le i^2p^2$ for $p\in[0,1]$, and consequently for our desired range $[1/n^{\frac{1}{4}-\delta'},\epsilon']$. This means that the negative seat probability essentially squares in each round, so that we expect $O(\log\log n)$ round until it is polynomially small. 

Formally, Lemma~\ref{lemma 5} (ii) will prove this. Based on Corollary~\ref{lemma 6}, $P_-^p\leq (1+\delta)\mathbb{E}[P_-^p]$ with probability at least $1-\exp(-\Theta(n\mathbb{E}[P_-^p]^2))$. Therefore, we can set $\underline{b}=1/n^{\frac{1}{4}-\delta'}$, $\overline{b}=\epsilon'$, $b=\epsilon'$, $I(\mathbb{E}[P_{-}^p])=\exp(-\Theta(n\mathbb{E}[P_-^p]^2))$, $R(p)=Kp^{\ell}$ for $K=i^2$ and $\ell=2$. 

Based on phase 1, we can select the constant $\epsilon'$ to be arbitrarily small. We set $\epsilon'=\frac{1}{2i^3}$ and $\delta=1$. Moreover, we have Claim~\ref{claim 3}, whose proof is given below. The choice of the exponents of $n$ is governed by the error bound of Lemma~\ref{lemma 5} (ii), see also Remark~\ref{rem:quarter} below.

\begin{claim}
\label{claim 3}
There exists some integer $m\le \log\log n$ such that $R^t(\epsilon')\in[1/n^{\frac{1}{4}-\delta'},\epsilon'] $ for $t<m$ and $\frac{1}{n^{\frac{1}{2}-2\delta'}}\le R^m(\epsilon')\le \frac{1}{n^{\frac{1}{4}-\delta'}}$.
\end{claim}
Now we can apply Lemma~\ref{lemma 5} (ii), which tells us that given $P_{-}(t_0)\le \epsilon'$, we have $P_{-}(t_0+m)\le R^m(\epsilon')\le 1/n^{\frac{1}{4}-\delta'}$ with probability at least

\begin{align*}
& 1-\sum_{t=1}^{m}\exp(-\Theta(n(R^t(\epsilon'))^2))\ge1-\sum_{t=1}^{m-1}\exp(-\Theta(n(\frac{1}{n^{\frac{1}{4}-\delta'}})^2))-\exp(-\Theta(n(\frac{1}{n^{\frac{1}{2}-2\delta'}})^2))=\\
& 1-(m-1)\exp(-\Theta(n^{\frac{1}{2}+2\delta'}))-\exp(-\Theta(n^{4\delta'}))
\end{align*}
where we used from Claim~\ref{claim 3} that $ R^m(\epsilon')\ge \frac{1}{n^{\frac{1}{2}-2\delta'}}$ and $R^t(\epsilon')\ge n^{\frac{1}{4}-\delta'}$ for $t<m$. Applying $m-1<\log\log n$ yields that this probability is at least $1-o(1)$.

Overall, given $P_+(0)\ge \alpha_i+\epsilon$, the process reaches the negative seat probability $1/n^{\frac{1}{4}-\delta'}$ in $\mathcal{O}(\log\log n)$ rounds.

\begin{remark}
\label{rem:quarter}
Note that for this argument the choice of $1/n^{\frac{1}{4}-\delta'}$ for $\delta'>0$ is crucial. If we set $\delta'=0$, then the error term $\exp(-\Theta(n^{4\delta'}))$ from above will be a constant. 
\end{remark}
\paragraph{Proof of Claim~\ref{claim 3}.} Firstly, we show by induction that $R^t(\epsilon')\le \epsilon'$ for $t\ge 0$. For the base case, $R^0(\epsilon')=\epsilon'$. As the induction hypothesis (I.H.) assume for some $t\ge 0$, $R^t(\epsilon')\le \epsilon'$. By applying $R^{t+1}(\epsilon')=(1+\delta)R(R^t(\epsilon'))$ and $R(p)=i^2p^2$ and replacing the values of $\delta=1$ and $\epsilon'=\frac{1}{2i^3}$, we have
\[
R^{t+1}(\epsilon')=(1+\delta)R(R^t(\epsilon'))\le 2i^2(R^t(\epsilon'))^2\stackrel{\text{I.H.}}{\le} (2i^2\epsilon') \epsilon'=\frac{\epsilon'}{i}\le \epsilon'.
\]
Furthermore, by part (iii) in Lemma~\ref{lemma 5} for $K=i^2$, $\ell=2$, $\delta=1$ , and $\epsilon'=\frac{1}{2i^3}$ we get 
\[
R^t(\epsilon')\le (K(1+\delta))^{\sum_{j=0}^{t-1}\ell^j}(\epsilon')^{\ell^t}=(2i^2)^{\sum_{j=0}^{t-1}2^j}(\epsilon')^{2^t}\le (2i^2)^{2^{t}}(\frac{1}{2i^3})^{2^t}=\frac{1}{i^{2^t}}.
\]
 This implies that $R^{t'}(\epsilon')\le \frac{1}{n}\le \frac{1}{n^{\frac{1}{4}-\delta'}}$ for $t'=\log\log n$. By combining this statement with $R^t(\epsilon')\le \epsilon'$ for $t\ge 0$ from above, we conclude that there exists some $m\le \log\log n$ such that $R^t(\epsilon')\in[1/n^{\frac{1}{4}-\delta'},\epsilon'] $ for $t<m$ and $R^m(\epsilon')\le 1/n^{\frac{1}{4}-\delta'}$. 

It is left to prove the lower bound on $R^m(\epsilon')$. By utilizing $R^{m-1}(\epsilon')\ge 1/n^{\frac{1}{4}-\delta'}$, we get
\[
R^m(\epsilon')=(1+\delta)R(R^{m-1}(\epsilon'))\ge 2i^2(\frac{1}{n^{\frac{1}{4}-\delta'}})^2\ge\frac{1}{n^{\frac{1}{2}-2\delta'}}.
\]
\paragraph{Phase 3.} Now, we show that from negative seat probability at most $1/n^{\frac{1}{4}-\delta'}$, the process reaches $\Sp$ in three more rounds a.a.s. 

Recall that $P_{-}^p=n_{-}^p/n$ and by Markov's inequality for a non-negative variable $X$ and $a>0$, $\Pr[X\ge a\cdot \mathbb{E}[X]]\le 1/a$. Therefore, we have
\[
\Pr[P_-^p\geq n^{\delta'}\mathbb{E}[P_-^p]]=\Pr[n_-^p\geq n^{\delta'}\mathbb{E}[n_-^p]]\leq n^{-\delta'}
\]
Furthermore, $\mathbb{E}[P_-^p]\leq i^2 p^2$, as discussed above. Thus, with probability at least $1-n^{-\delta'}$
\begin{equation}
\label{eq 7}
P_-^p< n^{\delta'}i^2p^2\le n^{2\delta'}p^2,
\end{equation}
assuming $n$ is large enough to guarantee $n^{\delta'}i^2\leq n^{2\delta'}$.

This implies that for $p\le 1/n^{\frac{1}{4}-\delta'}$, a.a.s. $P_-^p\le 1/n^{\frac{1}{2}-4\delta'}$. Let us apply Equation~(\ref{eq 7}) two more times. If $p\le 1/n^{\frac{1}{2}-4\delta'}$, then a.a.s. $P_-^p\le 1/n^{1-10\delta'}$. Finally, for $p\le 1/n^{1-10\delta'}$, we have a.a.s. $P_-^p\le 1/n^{2-22\delta'}$. Therefore, if $P_-(t_1)\le 1/n^{\frac{1}{4}-\delta'}$ for some $t_1\ge 0$, then a.a.s. $P_-(t_1+3)\le 1/n^{2-22\delta'}$, i.e., $n_-(t_1+3)\le 1/n^{1-22\delta'}<1$ for $\delta'=1/23$. Hence, the process reaches $\Sm$ in three more rounds a.a.s. \qed 

\paragraph{Tightness.} 
We claim that the upper bound of $\mathcal{O}(\log\log n)$ on the consensus time of the process in Theorem~\ref{theorem 5} is asymptotically tight. Assume that $a_{i}=1$ for some $i\ge 3$ and let $\alpha_i$ be the unique fixed point of $f_{i}(p)$ in $(0,1)$. Let us consider part (i), a very similar argument works for part (ii). We know that by starting from $P_-(0)=\epsilon'$ for some constant $0<\epsilon'<1-\alpha_i$ the process reaches $\Sp$ a.a.s. in $\mathcal
{O}(\log\log n)$ rounds. We argue that after $(\log_{i}\log_2 n)/2$ rounds a.a.s. there is still at least one negative individual. By Equation~(\ref{eq2}), we know that for $a_{i}=1$, 
\[
\mathbb{E}[P_-^p]=\sum_{j=\lceil i/2\rceil}^{i}{i \choose j}p^j(1-p)^{i-j}\geq p^{i}
\] 
for $p\in[0,1]$ and by Corollary~\ref{lemma 6}, the inequality $P_-^p\geq (1-\delta)\mathbb{E}[P_-^p]$ holds with probability at least $1-\exp(-\Theta(n\mathbb{E}[P_-^p]^2))$. Therefore, our set-up matches the conditions of Lemma~\ref{lemma 5} by considering $\underline{b}=0$, $\overline{b}=1$, $b=\epsilon'$, $\delta=1/2$, $I(\mathbb{E}[P_{-}^{p}])=\exp(-\Theta(n\mathbb{E}[P_-^p]^2))$, and $Q(p)=Kp^{\ell}$ for $K=1$ and $\ell=i$. It is only left to set a suitable $m$, which we do in Claim~\ref{claim 4}.

\begin{claim}
\label{claim 4}
Let $m=(\log_i\log_2n)/2$, then
\[
\omega(\frac{1}{n^{1/3}})\le Q^m(\epsilon')\le \cdots \le Q^{t-1}(\epsilon')\le Q^{t}(\epsilon')\le \cdots \le Q^{0}(\epsilon')=\epsilon'.
\]
\end{claim}
Therefore by Lemma~\ref{lemma 5} (i), starting from $P_-(0)=\epsilon'$, after $m=(\log_i\log_2n)/2$ rounds the negative seat probability is at least $Q^m(\epsilon')\ge \omega(1/n^{1/3})$ (i.e., the process does not reach $\Sp$) with probability at least $1-\sum_{t=1}^{m}\exp(-\Theta(n(Q^t(\epsilon'))^2))$. By applying Claim~\ref{claim 4},

\begin{align*}
& 1-\sum_{t=1}^{m}\exp(-\Theta(n(Q^t(\epsilon'))^2))\ge 1-\sum_{t=1}^{m}\exp(-\Theta(n(Q^m(\epsilon'))^2))\ge 1-m\cdot \exp(-\Theta(n\cdot\omega(1/n^{1/3})^2))=\\
& 1-\frac{\log_{i}\log_2 n}{2} \cdot\exp(-\omega(n^{1/3}))=1-o(1).
\end{align*}

\paragraph{Proof of Claim~\ref{claim 4}.}
We know that $Q^{0}(\epsilon')=Q(\epsilon')=\epsilon'$. Since $Q(p)=p^i$, by applying the definition of $Q^{t}(p)$ and setting $\delta=1/2$ we have for $t\ge 1$,
\[
Q^{t}(\epsilon')=(1-\delta)Q(Q^{t-1}(\epsilon'))=\frac{1}{2}(Q^{t-1}(\epsilon'))^i\le Q^{(t-1)}(\epsilon').
\] 
Furthermore, by part (iii) in Lemma~\ref{lemma 5} for $K=1$, $\ell=i$ and applying $\sum_{j=0}^{m-1}\ell^j\le \ell^{m}$, we get
\[
Q^{m}(\epsilon')\ge (K(1-\delta))^{\sum_{j=0}^{m-1}\ell^j}(\epsilon')^{\ell^m}\ge (\epsilon'/2)^{\ell^m}=(\epsilon'/2)^{i^m}.
\] 
Since $m=(\log_i\log_2n)/2$,
\[
Q^{m}(\epsilon')\ge (\epsilon'/2)^{i^m}=(\epsilon'/2)^{\sqrt{\log_2n}}=\omega((1/2)^{\log_2n^{1/3}})=\omega(1/n^{1/3}).
\]
\qed
 
\subsection{General Case: Rooms of Different Sizes}
\label{general-case}
In this section, we switch our attention to the general case with rooms of different sizes. In Theorem~\ref{theorem 7} we prove that the process exhibits a threshold behavior similar to the case of all rooms of the same size, provided that $\S\leq 16$. We should point out that this result perhaps holds also for $\S>16$, but this constraint originates from the limitation of our proof technique. 

Again, the main idea is to prove that $f(p)=\mathbb{E}[P_+^p]$ has a unique fixed point $\alpha$ in $(0,1)$ and then apply our concentration results from Section~\ref{ingredients}. However, our proof regarding the uniqueness of the fixed point only applies to $\S\le 16$. Therefore, to eliminate the constraint of $\S\leq 16$, one shall prove $f(p)$ has a unique fixed point in the interval of $(0,1)$ for any constant $\S$. This is an interesting problem by its own sake since $f(p)$ is actually the convex combination of binomial functions $f_i(p)$. 

\begin{theorem}
\label{theorem 7}
In the Galam model with all rooms of size 3 to 16, for any $\epsilon>0$

(i) $P_+(0)\geq\alpha+\epsilon$ results in $\Sp$.
 
(ii) $P_{+}(0)\leq\alpha-\epsilon$ results in $\Sm$
\\
in $\mathcal{O}(\log\log n)$ rounds a.a.s., where $\alpha$ is the unique fixed point of function $f(p)$ in $(0,1)$.
\end{theorem}
\begin{proof}
We first prove that function $h(p)=f(p)-p$ (see Definition~\ref{definition:function}) has a unique root $\alpha$ in the interval of $(0,1)$ and $h(p)<0$ for $p\in(0,\alpha)$ and $h(p)>0$ for $p\in(\alpha,1)$. Note this implies that function $f(p)$ has a unique fixed point in $(0,1)$.

Based on Equation~(\ref{eq 2'}) we know that the first derivative of $h(p)$ is equal to
\[
\frac{dh(p)}{dp}=\sum_{i=3}^{16}a_i(\lfloor \frac{i}{2}\rfloor+1) {i \choose \lfloor\frac{i}{2}\rfloor+1} p^{\lfloor\frac{i}{2}\rfloor} (1-p)^{i-\lfloor\frac{i}{2}\rfloor-1}-1.
\] 

This implies that $\frac{dh(p)}{dp}$ is equal to $-1$ for both $p=0$ and $p=1$. Thus, in a very close right neighborhood of $0$, we have $h(p)<h(0)=0$ and similarly in a very close left neighborhood of $1$, we have $h(p)>h(1)=0$. Hence, there exists at least one $\alpha \in (0,1)$ so that $h(\alpha)=0$.

Now, by exploiting the same argument which we had in the proof of Lemma~\ref{lemma 3}, i.e., applying Rolle's Theorem\cite{rosenlicht1968introduction}, we just have to show that $\frac{d^2h(p)}{dp^2}$ changes its sign at most once in the interval of $(0,1)$. We prove that for all $3\le i\le 16$,
\begin{enumerate}[label=(\alph*)]
\item $\frac{d^2h_i(p)}{dp^2}$ is positive for $0< p < 1/2$
\item $\frac{d^2h_i(p)}{dp^2}$ is negative for $2/3<p< 1$
\item $\frac{d^2h_i(p)}{dp^2}$ is monotonically decreasing for $1/2 \leq p \leq 2/3$.
\end{enumerate}
Since $\frac{d^2h(p)}{dp^2}=\sum_{i=3}^{16}a_i\frac{d^2h_i(p)}{dp^2}$, combining (a), (b), and (c) yields that there exists a point $p'\in(0,1)$ such that $\frac{d^2h(p)}{dp^2}$ is positive for $p\in(0,p')$ and negative for $p\in(p',1)$. 

By using Equation~(\ref{eq 2'}), we have
\[
\frac{d^2h_i(p)}{dp^2}=(\lfloor\frac{i}{2}\rfloor+1)(i-1){i \choose \lfloor\frac{i}{2}\rfloor+1}p^{\lfloor\frac{i}{2}\rfloor-1}(1-p)^{i-\lfloor\frac{i}{2}\rfloor-2}(\frac{\lfloor\frac{i}{2}\rfloor}{i-1}-p).
\]
Thus, the term $(\frac{\lfloor\frac{i}{2}\rfloor}{i-1}-p)$ determines the sign of $\frac{d^2h_i(p)}{dp^2}$. For $i\ge 3$, the maximum and the minimum value of $\frac{\lfloor\frac{i}{2}\rfloor}{i-1}$ are respectively $\frac{2}{3}$ (by setting $i=4$) and $\frac{1}{2}$ (by setting $i=3$). Therefore, for all $i\ge 3$, $\frac{d^2h_i(p)}{dp^2}$ is positive if $0< p <1/2$, and negative if $2/3<p< 1$.

It is left to prove the correctness of part (c). It suffices to show that $\frac{d^3h_i(p)}{dp^3}$ is negative for $1/2\leq p\leq 2/3$. The third derivative of $h_i(p)$ is equal to
\[
\frac{d^3h_i(p)}{dp^3}=(\lfloor\frac{i}{2}\rfloor+1)(i-1){i \choose \lfloor\frac{i}{2}\rfloor+1}p^{\lfloor\frac{i}{2}\rfloor-2}(1-p)^{i-\lfloor\frac{i}{2}\rfloor-3}\times
\] 
\[
((\lfloor\frac{i}{2}\rfloor-1)(1-p)(\frac{\lfloor\frac{i}{2}\rfloor}{i-1}-p)-p(i-\lfloor\frac{i}{2}\rfloor-2)(\frac{\lfloor\frac{i}{2}\rfloor}{i-1}-p)-p(1-p)).
\] 
Since all terms except the last one are positive for $i\ge 3$, the necessary and sufficient condition for $\frac{d^3h_i(p)}{dp^3}$ to be negative is that the following inequality holds.
\[(\lfloor\frac{i}{2}\rfloor-1)(1-p)(\frac{\lfloor\frac{i}{2}\rfloor}{i-1}-p)<p(i-\lfloor\frac{i}{2}\rfloor-2)(\frac{\lfloor\frac{i}{2}\rfloor}{i-1}-p)+p(1-p)
\]
For $i$ odd, this is obviously true because $i-\lfloor i/2 \rfloor -2=\lfloor i/2 \rfloor -1$ and $1-p \leq p$, and for even $i\leq 16$ also one can check by a simple computational software that it is correct. (It is worth to stress that it is exactly where our proof technique fails for $\S\ge 16$ since part (c) is not necessarily true for $\S>16$.)

Since Corollary \ref{lemma 6} and Lemma~\ref{lemma 5} also cover the setting of this theorem, the rest of the proof is very similar to the proof of Theorem \ref{theorem 5}. Phase 1 is directly applied. In phase 2 and phase 3, by utilizing $\mathbb{E}[P_-^p]\leq \S^2 p^2$ instead of $\mathbb{E}[P_-^p]\leq i^2 p^2$, the same arguments are followed. \qed
\end{proof}

\subsection{All Rooms of Size 2}
\label{room-size-2}
In this section, we consider the case of $a_2=1$, where all rooms are of size 2. In this setting, $\mathbb{E}[P^{p}_{+}]=p^2$ and function $f_2(p)=p^2$ is strictly smaller than $p$ for $p\in(0,1)$, which implies that $f_2(p)$ has no fixed point. Thus, we do not expect the process to exhibit a threshold behavior similar to the case of $i\ge 3$. Actually, one might expect the process to reach $\Sm$ by starting from any non-zero negative seat probability; we prove such result in Theorem \ref{theorem 2}. More precisely, we show that by starting from the negative seat probability $P_-(0)=\omega(1/n)$, the process reaches $\Sm$ after $\mathcal{O}(\log n)$ rounds a.a.s. Naturally, one might ask: Is the logarithmic bound on the consensus time tight or it can be replaced by the bound $\mathcal{O}(\log\log n)$ similar to $i\ge 3$? Does the statement hold for $P_-(0)=\mathcal{O}(1/n)$? We show that the logarithmic upper bound on the consensus time is tight. Furthermore, if $P_-(0)=c/n$ for a constant $c$, then the probability that all seats are positive initially is equal to
\[
(1-\frac{c}{n})^{n}\geq \frac{1}{4^c}
\]  
where we used the estimate $1-x\ge 4^{-x}$, which holds for $0\le x\le 1/2$. Thus, with a constant probability the process reaches $\Sp$ in the first round.

The main idea for the proof of Theorem~\ref{theorem 2} is again to apply Lemma~\ref{lemma 5} by selecting the suitable function $Q(p)$. However, this time instead of the concentration result from Corollary~\ref{lemma 6}, based on Azuma's inequality, we utilize the Chernoff bound~\cite{feller1968introduction} to provide Lemma~\ref{lemma 2}. Let us recall that based on the Chernoff bound, if $x_1, \cdots,x_n$ are independent Bernoulli random variables then for $X=\sum_{i=1}^{n}x_i$ and any constant $\delta>0$,
\[
\Pr[(1-\delta)\mathbb{E}[X]< X< (1+\delta)\mathbb{E}[X]]\geq 1-\exp(-\Theta(\mathbb{E}[X])).
\] 
\begin{lemma}
\label{lemma 2}
In the Galam model with $a_2=1$ and for any constant $\delta>0$,
\[
\Pr[(1-\delta)\mathbb{E}[P_{\pm}^p]<P_{\pm}^p<(1+\delta)\mathbb{E}[P_{\pm}^p]]\geq 1-\exp(-\Theta(n\mathbb{E}[P_{\pm}^p])).
\]
\end{lemma}
\begin{proof} 
We prove the statement for $P_+^p$ and the case of $P_-^p$ can be proven analogously. Consider an arbitrary labeling from $1$ to $n/2$ on the rooms. Assume that each seat is positive with probability $p$, independently. Define the Bernoulli random variable $x_i^p$, for $1\leq i\leq n/2$, to be one if and only if both seats in the $i$-th room are positive. Thus, by definition we have $n_+^p/2=\sum_{i=1}^{n/2}x_i^p$. Since random variables $x_1^p,\cdots,x_{n/2}^p$ are independent, using the Chernoff bound yields
\[
\Pr[(1-\delta)\mathbb{E}[n_+^p]<n_+^p<(1+\delta)\mathbb{E}[n_+^p]]\geq 1-\exp(-\Theta(\mathbb{E}[n_+^p])).
\]
Applying $n_{+}^p=nP_{+}^{p}$ finishes the proof. 
\qed 
\end{proof}
\begin{remark}
Note that the error probability, obtained by applying Azuma's inequality, in Corollary~\ref{lemma 6} is of form $\exp(-\Theta(n\mathbb{E}[P_{\pm}^p]^2))$ while in Lemma~\ref{lemma 2}, built on the Chernoff bound, the error probability is $\exp(-\Theta(n\mathbb{E}[P_{\pm}^p]))$. The latter gives a smaller error probability for $\mathbb{E}[P_{\pm}^p]=o(1)$, which is crucial for the proof of Theorem~\ref{theorem 2}. 
\end{remark}
\begin{theorem}
\label{theorem 2}
In the Galam model with $a_2=1$ and $P_-(0)=\omega(1/n)$, the process reaches $\Sm$ in $\mathcal{O}(\log n)$ rounds a.a.s.
\end{theorem}
\begin{proof}
We prove by starting from negative seat probability $P_-(0)=\omega(1/n)$ the process reaches the negative seat probability $1-\epsilon$ for an arbitrarily small constant $\epsilon>0$ in $\mathcal{O}(\log n)$ rounds (phase 1), then it reaches the negative seat probability at least $1-\frac{\log n}{\sqrt{n}}$ in $\mathcal{O}(\log\log n)$ rounds (phase 2). Finally, the process reaches $\Sm$ in two more rounds (phase 3).

Instead of working with the exact formula of $\mathbb{E}[P_{-}^p]$, we set a lower bound $Q(p)$ of the form $Kp^{\ell}$ on $\mathbb{E}[P_{-}^p]$, which is easier to handle but only holds for $p\le 1-\epsilon$, where $\epsilon$ is an arbitrarily small constant. By applying Lemma~\ref{lemma 5} and Lemma~\ref{lemma 2}, we show that the negative seat probability increases to $1-\epsilon$ a.a.s. This argument fails for $p>1-\epsilon$. To continue, we find an upper bound $R(p)$ on $\mathbb{E}[P_{+}^p]$ and again apply Lemma~\ref{lemma 5} and Lemma~\ref{lemma 2} to prove that the positive seat probability decreases to $\frac{\log n}{\sqrt{n}}$, in phase 2. However, we cannot rely on this argument to show that the process reaches $\Sm$ a.a.s. since the error probability provided by Lemma~\ref{lemma 2} is a constant for small values of $p$. However, Markov's inequality~\cite{feller1968introduction}, interestingly, yields a stronger tail bound for small values of $p$. This allows us to prove that the process reaches $\Sm$ in at most two more rounds, which is discussed in phase 3. 

\paragraph{Phase 1.}
First, we show that by starting from $P_-(0)=\omega(1/n)$, the process reaches the negative seat probability at most $1-\epsilon$, for an arbitrarily small constant $\epsilon>0$, after $\mathcal{O}(\log n)$ rounds a.a.s.

We want to apply Lemma~\ref{lemma 5} (i). Thus, first we have to determine $Q(p)$ for which $\mathbb{E}[P_{-}^p]\ge Q(p)$. By Equation~(\ref{eq2}) for $a_2=1$, we have
\[
\mathbb{E}[P^p_-]=2p(1-p)+p^2=2p-p^2.
\] 
Thus, $\mathbb{E}[P_-^p]\geq (1+\epsilon)p$ for $p\in[0,1-\epsilon]$, where $\epsilon>0$ is an arbitrarily small constant. We set $\underline{b}=0$, $\overline{b}=1-\epsilon$, $b=P_-(0)$, and $Q(p)=Kp^{\ell}$ for $K=(1+\epsilon)$ and $\ell=1$.

Furthermore by Lemma~\ref{lemma 2}, $P_-^p\geq (1-\delta)\mathbb{E}[P_-^p]$ with probability at least $1-\exp(-\Theta(n\mathbb{E}[P_{-}^{p}]))$ for any constant $\delta>0$. Thus, we can set $I(\mathbb{E}[P_{+}^p])=\exp(-\Theta(n\mathbb{E}[P_+^p]))$. We set constant $\delta>0$ to be sufficiently small such that $(1-\delta)(1+\epsilon)=1+\epsilon'$ for some small constant $\epsilon'>0$.

Having Claim~\ref{claim 5} (whose proof is given below) in hand, we are now ready to apply Lemma~\ref{lemma 5}.

\begin{claim}
\label{claim 5}
There is some $m\le \log_{1+\epsilon'} n$ such that $Q^m(b)\ge 1-\epsilon$ and for all $t<m$ we have $(1+\epsilon')^tb\le Q^t(b)\le 1-\epsilon$.
\end{claim}
Therefore, given $P_-(0)=\omega(1/n)$, we have $P_{-}(m)\ge Q^m(b)\ge 1-\epsilon$ with probability at least 
\[
1-\sum_{t=1}^{m}\exp(-\Theta(nQ^t(b))).
\]
By Claim~\ref{claim 5}, this probability is larger than

\[
1-\sum_{t=1}^{m}\exp(-\Theta(n(1+\epsilon')^t\omega(1/n)))\ge 1-\sum_{t=1}^{\log_{1+\epsilon'}n}\exp(-\omega(1)(1+\epsilon')^t)=1-o(1).
\]
Note that $\sum_{t=1}^{\log_{1+\epsilon'} n}\exp(-\omega(1)(1+\epsilon')^t)$ is clearly smaller than the geometric series
\[
\frac{1}{\omega(1)}\sum_{t=1}^{\log_{1+\epsilon'} n}\frac{1}{(1+\epsilon')^t}=o(1).
\]
\paragraph{Proof of Claim~\ref{claim 5}.}
By part (iii) in Lemma~\ref{lemma 5},
\begin{equation*}
\label{eq 3}
Q^t(b)\ge(K(1-\delta))^{\sum_{j=0}^{t-1}\ell^j}b^{\ell^t}=((1+\epsilon)(1-\delta))^{\sum_{j=0}^{t-1}1}b^{1^t}=(1+\epsilon')^tb.
\end{equation*} 
Thus, for $t=\lceil \log_{1+\epsilon'}\frac{1-\epsilon}{b}\rceil\le \log_{1+\epsilon'}n$, we have $Q^t(b)\ge (1-\epsilon)$. Furthermore,
\[
Q^t(b)=(1-\delta)Q(Q^{t-1}(b))=(1-\delta)(1+\epsilon)Q^{t-1}(b)=(1+\epsilon')Q^{t-1}(b)\ge Q^{t-1}(b).
\] 
The two last statements together imply that there exists some $m\le\log_{1+\epsilon'}n$ such that $Q^{m}(b)\ge 1-\epsilon$ and for all $t<m$, $Q^{t}(b)\le 1-\epsilon$. This finishes the proof of Claim~\ref{claim 5}. 

\paragraph{Phase 2.}
So far we proved that given $P_-(0)=\omega(1/n)$, the process reaches the negative seat probability at least $1-\epsilon$ for an arbitrarily small constant $\epsilon>0$ a.a.s in logarithmically many rounds, say $t_0$. Now, we show that if $P_+(t_0)\le \epsilon$, the process reaches the positive seat probability at most $\frac{\log n}{\sqrt{n}}$ in $\mathcal{O}(\log\log n)$ rounds. 

Similar to phase 1, we want to apply Lemma~\ref{lemma 5}, but this time part (ii). Thus, we have to find an upper bound $R(p)$ on $\mathbb{E}[P_+^p]$. By Equation~(\ref{eq 1}), we know $\mathbb{E}[P_+^p]=p^2$. Furthermore based on Lemma \ref{lemma 2}, $P_+^p\leq (1+\delta')\mathbb{E}[P_+^p]$ with probability at least $1-\exp(-\Theta(n\mathbb{E}[P_+^p]))$ for $\delta'>0$. We apply Lemma~\ref{lemma 5} (ii) for $\underline{b}=\frac{\log n}{\sqrt{n}}$, $\overline{b}=b=\epsilon$, and $R(p)=Kp^{\ell}$ for $K=1$ and $\ell=2$. Let us select $\delta'$ sufficiently small so that $(1+\delta')\epsilon=\epsilon^{\prime\prime}$ for some small constant $0<\epsilon^{\prime\prime}<1$.

Finally, we need Claim~\ref{claim 6}, which is proven below, to apply Lemma~\ref{lemma 5}.

\begin{claim}
\label{claim 6}
There exits some integer $m\le 2\log\log n$ such that for all $t<m$ we have $R^t(\epsilon)\in [\frac{\log n}{\sqrt{n}},\epsilon]$ and $\frac{\log^2n}{n}\le R^m(\epsilon)\le\frac{\log n}{\sqrt{n}}$.
\end{claim}
Therefore, given $P_+(t_0)\le \epsilon$, $P_+(t_0+m)\le R^{m}(\epsilon)\le \frac{\log n}{\sqrt{n}}$ with probability at least
\[
1-\sum_{t=1}^{m}\exp(-\Theta(nR^{t}(\epsilon)))\ge 1-(m-1)\exp(-\Theta(n\frac{\log n}{\sqrt{n}}))-\exp(-\Theta(n\frac{\log^2 n}{n}))=1-o(1).
\]
Overall, given $P_-(0)=\omega(1/n)$, the process reaches the positive seat probability at most $\frac{\log n}{\sqrt{n}}$ in $\mathcal{O}(\log n)$ rounds.
\begin{remark}
Note that for this argument the choice of $\frac{\log n}{\sqrt{n}}$ is crucial. If we set $p=\mathcal{O}(1/\sqrt{n})$, then the error probability from above will be a constant.
\end{remark}
\paragraph{Proof of Claim~\ref{claim 6}.}
Firstly, we show by induction that $R^t(\epsilon)\le \epsilon$ for $t\ge 0$. For the base case, $R^0(\epsilon)=\epsilon$. As the induction hypothesis (I.H.) assume for some $t\ge 0$, $R^t(\epsilon)\le \epsilon$. By applying $R^{t+1}(\epsilon)=(1+\delta')R(R^t(\epsilon))$ and $R(p)=p^2$, we have
\[
R^{t+1}(\epsilon)=(1+\delta')R(R^t(\epsilon))\stackrel{\text{I.H.}}{\le} (1+\delta')\epsilon^2=\epsilon^{\prime\prime}\epsilon\le \epsilon.
\]

Furthermore, by part (iii) in Lemma~\ref{lemma 5}, we have 
\[
R^{t}(\epsilon)\le (1+\delta')^{\sum_{j=0}^{t-1}2^j}\epsilon^{2^{t}}\le ((1+\delta')\epsilon)^{2^t}=(\epsilon^{\prime\prime})^{2^t}
\]
which is smaller than $\frac{\log n}{\sqrt{n}}$ for $t=2\log\log n$. By combining this statement with $R^t(\epsilon)\le \epsilon$ for $t\ge 0$ from above, we conclude that there exists some $m\le 2\log\log n$ so that for all $t<m$ $R^{t}(\epsilon)\in [\frac{\log n}{\sqrt{n}},\epsilon]$ and $R^{m}(\epsilon)\le \frac{\log n}{\sqrt{n}}$. 

It is left to prove the lower bound on $R^m(\epsilon)$. By applying $R^{m-1}(\epsilon)\ge \frac{\log n}{\sqrt{n}}$, we get
\[
R^{m}(\epsilon)=(1+\delta')(R^{m-1}(\epsilon))^2\ge (1+\delta')\frac{\log^2n}{n}\ge \frac{\log^2 n}{n}
\]
which finishes the proof of Claim~\ref{claim 6}.

\paragraph{Phase 3.}
It remains to show that from positive seat probability at most $\frac{\log n}{\sqrt{n}}$, the process reaches $\Sm$ in two more rounds.

Recall that $P_{+}^p=n_{+}^p/n$ and by Markov's inequality for a non-negative variable $X$ and $a>0$, $\Pr[X\ge a\cdot \mathbb{E}[X]]\le 1/a$. Therefore, we have
\[
\Pr[P_+^p\geq \log n\cdot\mathbb{E}[P_+^p]]=\Pr[n_+^p\geq \log n\cdot\mathbb{E}[n_+^p]]\leq \frac{1}{\log n}.
\]
Furthermore, $\mathbb{E}[P_+^p]=p^2$. Thus with probability at least $1-1/\log n$, we have $P_+^p\le \log n\cdot p^2$. This implies that for $p\le \frac{\log n}{\sqrt{n}}$, a.a.s. $P_+^p\le \log n\cdot \frac{\log^2n}{n}=\frac{\log^3 n}{n}$. By applying the same argument one more time we get that if $p\le \frac{\log^3 n}{n}$, then a.a.s. $P_+^p\le \frac{\log^7}{n^2}$. Therefore, if $P_+(t_1)\le \frac{\log n}{\sqrt{n}}$ for some $t_1\ge 0$, then a.a.s. $P_+(t_1+2)\le \frac{\log^7 n}{n^2}$, i.e., $n_+(t_1+2)\le \frac{\log^7 n}{n}<1$. Hence, the process reaches $\Sm$ in two more rounds a.a.s. \qed
\end{proof} 

\paragraph{Tightness.} 
We claim that the bound $\mathcal{O}(\log n)$ on the consensus time of the process in Theorem~\ref{theorem 2} is asymptotically tight. We prove that for $a_2=1$, from $P_-(0)=\frac{1}{\sqrt{n}}$, after $\frac{\log n}{6}$ rounds a.a.s. the process still has not reached $\Sm$.

We know that $\mathbb{E}[P_-^p]=2p-p^2\leq 2p$ for any $p\in[0,1]$ and by Lemma~\ref{lemma 2}, $P_-^p\leq (1+\delta)\mathbb{E}[P_-^p]$ with probability at least $1-\exp(-\Theta(n\mathbb{E}[P_-^p]))$. Therefore, our set-up matches the conditions of Lemma~\ref{lemma 5} (ii) by considering $\underline{b}=0$, $\overline{b}=1$, $b=\frac{1}{\sqrt{n}}$, $\delta=1$ and $R(p)=Kp^{\ell}$ for $K=2$ and $\ell=1$. It is only left to set a suitable $m$, which we do in Claim~\ref{claim 7}.

\begin{claim}
\label{claim 7}
Let $m=\frac{\log_2 n}{6}$, then
\[
\frac{1}{\sqrt{n}}=R^{0}(b)\le R^{1}(b)\le \cdots \le R^{m-1}(b)\le R^{m}(b)\le\frac{1}{n^{1/6}}.
\]
\end{claim}
Therefore, given $P_-(0)=\frac{1}{\sqrt{n}}$, we have $P_-(m)\le R^m(b)\le \frac{1}{n^{1/6}}$ with probability at least
\[
1-\sum_{t=1}^{m}\exp(-\Theta(nR^m(b)))\ge 1-m\exp(-\Theta(n\frac{1}{\sqrt{n}}))=1-o(1).
\]
Note that $P_-(m)\le \frac{1}{n^{1/6}}$ is equivalent to $n_-(m)\le n^{5/6}$. Hence, after $m=\frac{\log_2 n}{6}$ rounds a constant fraction of seats are still positive a.a.s.

\paragraph{Proof of Claim~\ref{claim 7}.}
By definition, for any $t\ge 1$
\[
R^t(b)=(1+\delta)R(R^{t-1}(b))=2(2R^{t-1}(b))\ge R^{t-1}(b)
\]
Moreover,
\[
R^m(b)\le (K(1+\delta))^{\sum_{j=0}^{m-1}\ell^j}b^{\ell^m}\le 2^{2m}\frac{1}{\sqrt{n}}=\frac{n^{1/3}}{n^{1/2}}=\frac{1}{n^{1/6}}.
\]
Thus,
\begin{equation*}
\label{eq 4}
\frac{1}{\sqrt{n}}=R^{0}(b)\le R^{1}(b)\le \cdots \le R^{m-1}(b)\le R^{m}(b)\le\frac{1}{n^{1/6}}.
\end{equation*}

\subsection{Threshold Value and Room Sizes}
\label{threshold-room}
The distribution of the room sizes is the main determining parameter of the Galam model. To have a better insight into the behavior of the model, it is essential to understand how alternating the room sizes affects the threshold behavior of the process. In the rest of this section, we aim to address this question. 

Recall from Theorem~\ref{theorem 5} that if $a_i=1$ (i.e., all rooms are of the same size $i$) for $i\ge 3$, then the process exhibits a threshold behavior at $\alpha_i$, the unique fixed point of $f_i(p)$ in $(0,1)$. If $i$ is odd, there is no tie-breaking and thus the threshold value is equal to $1/2$. In Theorem~\ref{theorem 9}, we prove that for even $i$ $\alpha_{i+2}<\alpha_i$, that is the threshold for the positive opinion to win the debate goes down if rooms get larger. This is plausible, since ties (that benefit the negative opinion) become less likely.

Let us first discuss the following simple lemma, which we apply to prove Theorem~\ref{theorem 9}. Assume that we flip a biased coin, which is in favor of head, $i$ times independently and we will win if more than half of the flips are head. Lemma \ref{lemma 1} demonstrates that the probability that we win when we flip the coin $i+2$ times is higher than if we do it $i$ times.

\begin{lemma}
\label{lemma 1}
For $f_i(p)=\sum_{j=\lfloor\frac{i}{2}\rfloor+1}^{i}{i \choose j}p^{j}(1-p)^{i-j}$ and $i\in \mathbb{N}$

(i) $f_i(p)<f_{i+2}(p)$ if $1/2<p<1$.

(ii) $f_i(p)>f_{i+2}(p)$ if $0<p<1/2$ and $i$ odd.
\end{lemma}
\begin{proof} 
We discuss the case that $i$ is odd and $1/2<p<1$; the case that $i$ is even and part (ii) can be proven analogously. $f_i(p)$ is equal to the probability that we flip a biased coin (with head probability $p$ and tail probability $1-p$) $i$ times independently and more than half of the trials are head. Assume by the probability of win with $i$ trials, we mean the probability that more than half of $i$ independent trials are head. Then, it is sufficient to show that the probability of win with $i+2$ trials is more than the probability of win with $i$ trials.

Assume random variables $H$ and $T$ denote the number of heads and tails, respectively, after $i$ trials. Now, assume $q_1:=\Pr[H\geq T+3]$, $q_2:=\Pr[H=T+1]$, and $q_3:=\Pr[H=T-1]$. By considering that $i$ is odd, the probability of win with $i$ trials is $q_1+q_2$. In other words, $f_i(p)=q_1+q_2$. Furthermore, the probability of win with $i+2$ trials is equal to $q_1+q_2(1-(1-p)^2)+q_3p^2$ which implies $f_{i+2}(p)=q_1+q_2(1-(1-p)^2)+q_3p^2$. Thus to prove $f_i(p)<f_{i+2}(p)$, it suffices to show that $q_2(1-p)^2<q_3p^2$ which is equivalent to
\[
{i \choose \lfloor \frac{i}{2}\rfloor} p^{\lfloor \frac{i}{2}\rfloor+1} (1-p)^{\lfloor \frac{i}{2}\rfloor} (1-p)^2<{i \choose \lfloor \frac{i}{2}\rfloor} p^{\lfloor \frac{i}{2}\rfloor} (1-p)^{\lfloor \frac{i}{2}\rfloor+1} p^2 \Leftrightarrow 1-p < p
\]
and it is true by $p>1/2$. \qed
\end{proof}
\begin{remark}
Note that part (ii) of Lemma~\ref{lemma 1} does not hold for even $i$. Assume that $p$ is almost $1/2$, then for $i=2$, $f_i(p)\simeq 4/16$ and for $i=4$, $f_i(p)\simeq 5/16$.
\end{remark}
\begin{theorem}
\label{theorem 9}
Let $\alpha_i$ and $\alpha_{i+2}$ be the threshold values respectively for $a_i=1$ and $a_{i+2}=1$ for some even $i\ge 3$, then $\alpha_i>\alpha_{i+2}$.
\end{theorem}
\begin{proof} 
Based on Theorem \ref{theorem 5} the threshold values $\alpha_i$ and $\alpha_{i+2}$ are respectively the unique root of $h_i(p)$ and $h_{i+2}(p)$ in the interval of $(0,1)$. Furthermore, $h_i(p)$ ($h_{i+2}(p)$, respectively) is negative from $0$ to $\alpha_i$ ($\alpha_{i+2}$, respectively) and positive from $\alpha_i$ ($\alpha_{i+2}$, respectively) to $1$. Therefore, to prove $\alpha_i>\alpha_{i+2}$, it is sufficient to prove the two following statements.

\textit{(i)} $h_i(p)<0$ and $h_{i+2}(p)<0$ for $0<p\leq 1/2$. (This implies that $\alpha_i, \alpha_{i+2} >1/2$.) 

\textit{(ii)} $h_{i+2}(p)>h_{i}(p)$ for $1/2<p<1$. (This is deduced directly from Lemma \ref{lemma 1}.)

To prove (i), it suffices to show that $h_i(1/2)<0$ and $h_{i+2}(1/2)<0$. We show $h_i(1/2)<0$; the case of $h_{i+2}(1/2)<0$ is proven similarly.
\[
h_i(\frac{1}{2})=\frac{1}{2^i}\sum_{j=i/2+1}^{i} {i \choose j}-\frac{1}{2}=\frac{1}{2^i}(2^{i-1}-\frac{{i \choose \frac{i}{2}}}{2})-\frac{1}{2}=-\frac{{i \choose \frac{i}{2}}}{2^{i+1}}<0. 
\]\qed

\end{proof}

\section{Future Research Directions}
\label{conclusion}
As we mentioned in Section \ref{Introduction}, after the introduction of the Galam model \cite{galam2002minority}, a series of extensions (see e.g.~\cite{galam2004contrarian,galam2005heterogeneous,gekle2005opinion,bischi2010binary,qian2015activeness}) have been suggested to make the Galam model more realistic to describe social situations. We shortly discuss some prior experimental and analytical results concerning these extensions. Naturally, a potential direction for future research is to approach these variants from a more theoretical angle. 

In \cite{galam2005heterogeneous}, a generalized version of the model was presented, where in case of a tie each individual becomes positive with fixed probability $k$ and negative with probability $1-k$ independently. One can see that in this model, we have
\[
\mathbb{E}[P_{+}^p]=\sum_{i=1}^{\S}a_i \sum_{j=\lfloor\frac{i}{2}\rfloor+1}^{i} {i \choose j} p^{j}(1-p)^{i-j}+k\sum_{i=1}^{\lfloor \S/2\rfloor}a_{2i}{2i \choose i}p^i (1-p)^i.
\]
The Galam model is a special case of the aforementioned model for $k=0$. It is conjectured~\cite{galam2005heterogeneous} that this generalized variant also exhibits a threshold behavior. We claim our concentration result from Corollary \ref{lemma 6} can be simply extended to this setting; however, arguing the uniqueness of the fixed point remains as an open problem for future research.

Gekle et al. \cite{gekle2005opinion} investigated the case of three opinions A, B, and C instead of two but in a very specific setting of all rooms of size 3. This leads to the possibility of ties, when each member of the group has a different opinion. They considered the probabilities $P_A$, $P_B$, and $1-P_A-P_B$ of resolving the tie in favor of A, B, and C respectively and tried to analyze the behavior of the model in this setting by computer simulations. 

Furthermore, Galam~\cite{galam2004contrarian} analytically studied the effect of adding the contrarian individuals into the model. A contrarian is an individual who adopts the choice opposite to the prevailing choice of others whatever this choice is.

According to the room-wise majority model, which is the base of the Galam model, each individual participates at every update of an opinion interaction. While the scheme gives everyone the same chance to influence others, in reality, social activity and influence vary considerably from one individual to another. To account for such a feature, Qian et al. \cite{qian2015activeness} introduced a new individual attribute of \textit{activeness} which makes some individuals more inclined than others at engaging in local discussions.

Galam and Jacobs~\cite{galam2007role} introduced another variant of the model, which simulates a society that includes inflexible members. At contrast to \textit{floaters}, the individuals who do flip their opinion to follow the local majority, \textit{inflexibles} keep their opinion always unchanged. The authors analytically investigate the inflexible effect in the specific setting of all rooms of size three. Let us define their model more formally. Similar to the original Galam model, assume there are $r_i$ rooms with $i$ seats for $2\le i \leq \S$ and assume $\sum_{i=1}^{\S}a_i=1$, where $a_i:=\frac{ir_i}{n}$, which implies there are exactly $n$ seats. Assume $P_{\pm}(t)$ for $t \geq 0$ denotes the positive/negative seat probability in round $t$. The process starts from positive seat probability $P_{+}(0)$, negative seat probability $P_{-}(0)$, positive inflexible seat probability $a$, and negative inflexible seat probability $b$ with $P_{+}(0)+P_{-}(0)=1$, $P_+(0)\geq a$, and $P_-(0)\geq b$. In round $t\geq 1$, each seat will be independently positive inflexible, positive floater, negative inflexible, and negative floater respectively with probability
$a$, $P_{+}(t-1)-a$, $b$, and $P_{-}(t-1)-b$.

Once each seat has been assigned an opinion, we proceed as in the room-wise majority model to determine the number of positive seats, with the exception that inflexibles keep their opinion unchanged. 

In this process eventually the positive seat probability decreases to $a$ (we say the victory of negatives) or negative seat probability decreases to $b$ (victory of positives). Galam and Jacobs~\cite{galam2007role} conjectured that this process shows a threshold behavior; i.e., for a threshold value $\alpha$ in $(a,1-b)$, $P_{+}(0)<\alpha$ results in the victory of negatives and $P_{+}(0)>\alpha$ outputs the victory of positives.

Let us define random variable $P_{\pm}^p$ similar to the Galam model, which tells us how the positive/negative seat probability evolves after one round of the process. We claim that 
\[
\mathbb{E}[P_+^p]=\sum_{i=1}^{\S}a_i\Big(\sum_{j=\lfloor\frac{i}{2}\rfloor+1}^{i}{i-1 \choose j-1}p^j(1-p)^{i-j}+
\]
\[
\sum_{j=\lfloor\frac{i}{2}\rfloor+1}^{i}{i-1 \choose j}(1-p-b)p^j(1-p)^{i-j-1}+\sum_{j=1}^{j=\lfloor\frac{i}{2}\rfloor}{i-1 \choose j-1}ap^{j-1}(1-p)^{i-j}\Big).
\]
For a fixed seat in a room of size $i$, the first term in the aforementioned sum stands for the probability that the seat is positive, floater or inflexible, and the dominant opinion in the room is positive. The second term is the probability that the seat is negative floater, however the dominant opinion is positive, and the third term corresponds to the state that the seat is positive inflexible but negatives are not less than positives in the room. Clearly, these are the only possible scenarios in which the output of a seat is positive.
 
Our result regarding the concentration of $P_{\pm}^p$ in Corollary \ref{lemma 6} can be simply extended to this setting. Therefore, the next step might be to prove that the function corresponding to $\mathbb{E}[P_+^p]$ has a unique fixed point in the interval of $(a,1-b)$.

Finally, we shortly discuss the relation between the room sizes distribution and the consensus time of the process in the Galam model. Roughly speaking, rooms of larger size accelerate the process because the opinion of a seat is decided based on the dominant opinion in a larger sample of the community. For instance, in the extreme case of just one room of size $n$, the process reaches $\Sp$ or $\Sm$ in one round, but for all rooms of size two, the process might need logarithmically many rounds in expectation to reach $\Sp$ or $\Sm$ (see Section~\ref{room-size-2}). To address this observation, let us compare two special cases of $a_i=1$ and $a_{i+2}=1$ for some odd $i$. The comparison of the consensus times in this setting is sensible since there is no tie-breaking and thus the threshold value in both cases is $1/2$.
Suppose random variables $P_{+,1}(t)$ and $P_{+,2}(t)$ denote the positive seat probability after $t$ rounds respectively for $a_i=1$ and $a_{i+2}=1$. Define $P_{+,1}^p$ and $P_{+,2}^p$ similar to $P_{+}^p$.
Based on Lemma \ref{lemma 1}, $\mathbb{E}[P_{+,1}^p]<\mathbb{E}[P_{+,2}^p]$ for $1/2<p<1$ and $\mathbb{E}[P_{+,2}^p]<\mathbb{E}[P_{+,1}^p]$ for $0<p<1/2$. Therefore, by starting from the same initial positive seat probability, the process converges to $\Sm$ or $\Sp$ in the case of $a_{i+2}=1$ faster than $a_i=1$, in expectation. 

It would be interesting to analyze the relation between the room sizes distribution and the consensus time more formally and in a more general framework in future work.

A concrete open problem arising from this work is to remove the limitation of $\S\leq 16$ on the room size in Theorem~\ref{theorem 7}. For this, it would be sufficient to prove that the polynomial 
\[
\mathbb{E}[P_+^p]= \sum_{i=1}^{\S}a_i \sum_{j=\lfloor\frac{i}{2}\rfloor+1}^{i} {i \choose j} p^{j}(1-p)^{i-j}
\]
from Equation~(\ref{eq 1}) has a unique fixed point whenever $\sum_{i=1}^{\S} a_i=1$ and
$a_i\geq 0$ for $1\leq i\leq L$. We have proved this if $L\leq 16$, or if $a_i=1$ for some $i\geq 3$.

\paragraph{Acknowledgment.}We wish to thank Serge Galam for referring to some prior works.
\bibliographystyle{acm} %
\bibliography{ref}
\end{document}